\documentclass{jac}

\usepackage[english]{babel}
\usepackage[utf8]{inputenc}
\usepackage{amsmath,amsthm,amssymb}
\usepackage{xcolor}
\usepackage{tikz}
\usetikzlibrary{decorations.pathmorphing,positioning,patterns}
\usepackage[dessinetikz,dessinempost]{optional}
\newcommand\sshift[1]{\ensuremath{\mathbf{#1}}}

\newcommand\plan[1][\groupe{\Z^2}]{\ensuremath{\mathbf{#1}}}

\newcommand\couleur[1]{\mathbf{#1}}
\newcommand\motif[1]{\ensuremath{\mathbf{#1}}}

\newcommand\domain[1]{\ensuremath{\mathcal{D}_{#1}}}

\newcommand\enspav{\mathbf{T}}

\newcommand\espace[1][\plan]{\ensuremath{\couleur{Q}^{#1}}}

\newcommand\conf[1]{#1}

\newcommand\shiftalone[1]{\ensuremath{\sigma_{#1}}}
\newcommand\shift[2]{\ensuremath{\shiftalone{#1}(\conf{#2})}}

\newcommand\ForbPat{\ensuremath{\mathcal{F}}}
\newcommand\Tilings[1][\ForbPat]{\ensuremath{\enspav_{#1}}}

\newcommand\latin[1]{\emph{#1}}

\newcommand\ie{\latin{i.e.}, }
\newcommand\eg{\latin{e.g.}, }
\newcommand\Z{\ensuremath{\mathbb{Z}}}
\newcommand\N{\ensuremath{\mathbb{N}}}

\newcommand\qpfunc[1]{\ensuremath{\mathcal{Q}_{#1}}}
\newcommand\qppat[2]{\ensuremath{n_{(\motif{#1},#2)}}}
\newcommand\qppatloc[2]{\ensuremath{m_{(\motif{#1},#2)}}}
\newcommand\fourpat[3]{
\ensuremath{\motif{#1}_{|[-#2-#3;-#2]^2}},
\ensuremath{\motif{#1}_{|[-#2-#3;-#2]\times[#2;#2+#3]}},
\ensuremath{\motif{#1}_{|[#2;#2+#3]\times[-#2-#3;-#2]}},
\ensuremath{\motif{#1}_{|[#2;#2+#3]^2}}}

\newcommand\mgreen{green!75!black!50}

\newcommand\mblue{blue!100!black!20}

\begin{document}

\title[Computing (or not) Quasi-periodicity Functions of Tilings]{Computing (or not)\\ Quasi-periodicity Functions of Tilings}

\author%[lif]
{A. Ballier}{Alexis Ballier}
\author%[lif]
{E. Jeandel}{Emmanuel Jeandel}
\address%[lif]
{Laboratoire d'Informatique Fondamentale de Marseille\\
CMI, 39 rue Joliot-Curie\\
F-13453 Marseille Cedex 13, France}
\email[A. Ballier]{alexis.ballier@lif.univ-mrs.fr}
\email[E. Jeandel]{emmanuel.jeandel@lif.univ-mrs.fr}

\thanks{Both authors are partly supported by ANR-09-BLAN-0164. A. Ballier has
been partly supported by the Academy of Finland project 131558. We thank Pierre
Guillon for discussions that lead to the constructions provided in
Section~\ref{sec:lowbound}.}

\begin{abstract}\noindent
We know that tilesets that can tile the plane always admit a quasi-periodic
tiling~\cite{birkmin,Durand99}, yet they hold many uncomputable
properties~\cite{berger,Hanf74,Myers74,simpsonmedv}.
The quasi-periodicity function is one way to measure the regularity of a
quasi-periodic tiling.
We prove that the tilings by a tileset that admits only quasi-periodic tilings
have a recursively (and uniformly) bounded quasi-periodicity function.
This corrects an error from~\cite[theorem~$9$]{CervelleD04} which stated the
contrary.
Instead we construct a tileset for which any quasi-periodic tiling has a
quasi-periodicity function that cannot be recursively bounded.
We provide such a construction for $1-$dimensional effective subshifts and
obtain as a corollary the result for tilings of the plane \latin{via} recent
links between these objects~\cite{mathieunathaliesoficeff,drseffsofic}.
\end{abstract}

\maketitle

Tilings of the discrete plane as studied nowadays have been introduced by Wang
in order to study the decidability of a subclass of first order
logic~\cite{wangpatternrecoI,wangpatternrecoII,BorgerGG1997}.
After Berger proved the undecidability of the domino problem~\cite{berger},
interest has grown for understanding how complex are these simply defined
objects~\cite{Hanf74,Myers74,complextilingsjsl,CervelleD04}.
Despite being able to have complex tilings, any tileset that can tile the plane
admits a \emph{quasi-periodic} tiling~\cite{birkmin,Durand99}; roughly speaking, a
quasi-periodic tiling is a tiling in which every finite pattern can be found in
any sufficiently large part of the tiling. It is therefore natural to define the
quasi-periodicity function of a quasi-periodic tiling: it associates to an
integer $n$ the minimal size in which we are certain to find any pattern of
size $n$~\cite{Durand99,CervelleD04}.
This is one way to measure the complexity of a quasi-periodic tiling and, to
some extent, of a tileset $\tau$ since $\tau$ must admit at least one quasi-periodic tiling.
We start by proving in Section~\ref{sec:bound} that tilings by tilesets that
admit only quasi-periodic tilings have a recursively (and uniformly) bounded
quasi-periodicity function (Theorem~\ref{thm:main}).
Remark that there exists non-trivial tilesets that admit only quasi-periodic
tilings~\cite{robinson,moz,Ollinger08} and that the property of having only such
tilings can be reduced to the domino problem~\cite{berger,robinson} and is thus
undecidable\footnote{Take a tileset $\tau_u$ that admits only one uniform
tiling (and thus only quasi-periodic tilings), a tileset $\tau_f$ that admits
non quasi-periodic tilings (\eg a fullshift on $\left\{0,1\right\}$) then it is
clear that $\left(\tau\times\tau_f\right)\cup\tau_u$ admits only quasi-periodic
tilings if and only if $\tau$ does not tile the plane.}.

With the aim to study discretization of dynamical systems, $1-$dimensional
subshifts have been extensively studied in symbolic
dynamics~\cite{SymbDyn,marcuslind}. Quasi-periodic tilings correspond to almost
periodic sequences~\cite{Hedlund69} or uniformly recurrent sequences in this
context.
Again, the existence of complex uniformly recurrent sequences has been
shown~\cite{apseq}.
In Section~\ref{sec:lowbound} we show, given a partial recursive function
$\varphi$, how to construct an effective subshift in which every uniformly
recurrent configuration has a quasi-periodicity function greater than $\varphi$
where it is defined (Theorem~\ref{thm:1deffnonbound}).
This allows us to correct the error from~\cite{CervelleD04} as we obtain as a
corollary (using recent links between tilings and effective $1-$dimensional
subshifts~\cite{mathieunathaliesoficeff,drseffsofic}) that there exists tilesets
for which no quasi-periodic tiling can have a quasi-periodicity function that is
recursively bounded (Theorem~\ref{thm:qpfuncunbound}).

\section{Definitions}

A \emph{configuration} is an element of $\espace$ where $\couleur{Q}$ is a
finite set or, equivalently, a mapping from $\plan$ to $\couleur{Q}$.
A \emph{pattern} $\motif{P}$ is a function from a finite domain
$\domain{P}\subseteq\plan$ to $\couleur{Q}$.
The \emph{shift} of vector $v$ ($v\in\plan$) is the function denoted by
$\sigma_v$ from $\espace$ to $\espace$ defined by $\shift{v}{c}(x) =
\conf{c}(v+x)$.
A pattern $\motif{P}$ \emph{appears} in a configuration $\conf{c}$ (denoted
$\motif{P}\in\conf{c}$) if there exists $v\in\plan$ such that
$\shift{v}{c}_{|\domain{P}}=\motif{P}$.
Similarly, we can define the shift of vector $v$ of a pattern $\motif{P}$ by
the function $\shift{v}{\motif{P}}(x)=\motif{P}(v+x)$; then we can say that a
pattern $\motif{P}$ appears in another pattern $\motif{M}$ if there exists
$v\in\plan$ such that $\shift{v}{\motif{M}}_{|\domain{P}}=\motif{P}$ and denote
it by $\motif{P}\in\motif{M}$. We use the same vocabulary and notations for both
notions of shift and appearance but there should not be any confusion since
configurations are always denoted by lower case letters and patterns by upper
case letters.

Given a finite set of colors $\couleur{Q}$, a \emph{tileset} is defined by a
finite set of patterns $\ForbPat$; we say that a configuration $\conf{c}$ is a
\emph{valid tiling} for $\ForbPat$ if none of the patterns of $\ForbPat$ appear
in $\conf{c}$.
We denote by $\Tilings$ the set of valid tilings for $\ForbPat$. If $\Tilings$
is non-empty we say that $\ForbPat$ can tile the plane.
A set of configurations $\sshift{T}$ is said to be a \emph{set of tilings} if
there exists some finite set of patterns $\ForbPat$ such that
$\sshift{T}=\Tilings$.
This notion of set of tilings corresponds to subshifts of finite
type~\cite{marcuslind,lind}. When we impose no restriction on $\ForbPat$ these
are subshifts and when $\ForbPat$ is recursively enumerable we say that
$\Tilings$ is an \emph{effective subshift} (see, \eg
\cite{compuniveffsymb,hochmanrecsshift,hochmanuniv,mathieunathaliesoficeff,drseffsofic}).

A periodic configuration $\conf{c}$ is a configuration such that the set
$\left\{\shift{v}{c}, v\in\plan\right\}$ is finite.
It is well known (since Berger~\cite{berger}) that there exists
tilesets that do not admit a periodic tiling but can still tile the plane.
On the other hand, quasi-periodicity is the correct regularity notion if we
always want a tiling with this property.
Periodic configurations are quasi-periodic but the converse is not true.
Several characterizations of quasi-periodic configurations
exist~\cite{Durand99}, we give one here that we use for the rest of the paper.

\begin{defi}[Quasi-periodic configuration]
A configuration $\conf{c}\in\espace$ is quasi-periodic if any pattern that
appears in $\conf{c}$ appears in any sufficiently large pattern of $\conf{c}$.

More formally, if a pattern $\motif{P}$ appears in $\conf{c}$ then there exists
$n\in\N$ such that for every pattern $\motif{M}$ defined on $[-n;n]^2$ that
appears in $\conf{c}$, $\motif{P}$ appears in $\motif{M}$.

We denote by $\qppat{P}{\conf{c}}$ the smallest such $n$ for finding a
pattern $\motif{P}$ in the quasi-periodic configuration $\conf{c}$.
\end{defi}

\begin{theorem}[\cite{birkmin,Durand99}]
Any non-empty set of tilings contains a quasi-periodic configuration.
\end{theorem}

For an integer $n$, the set of patterns defined on a square domain $[-n;n]^2$ is
finite, it is therefore natural to define the quasi-periodicity function of a
quasi-periodic configuration.

\begin{defi}[Quasi-periodicity function]
The quasi-periodicity function of a quasi-periodic configuration $\conf{c}$,
denoted by $\qpfunc{c}$, is the function from $\N$ to $\N$ that maps a given integer
$n$ to the smallest integer $m$ such that any pattern of domain $[-n;n]^2$ that
appears in $\conf{c}$ appears in any pattern of $\conf{c}$ of domain $[-m;m]^2$.
\[
\qpfunc{c}(n) = \max \left\{\qppat{P}{\conf{c}}, \motif{P}\in\conf{c},\domain{P}=[-n;n]^2\right\}
\]
\end{defi}

The function $\qpfunc{c}$ measures in some sense the complexity of the
quasi-periodic configuration $\conf{c}$: the faster it grows, the more complex
$\conf{c}$ is.
Since one can construct tilesets whose tilings have many uncomputable
properties (\eg such that every tiling is uncomputable as a function from
$\plan$ to $\couleur{Q}$~\cite{Hanf74,Myers74} or such that every pattern that
appears in a tiling has maximal Kolmogorov complexity~\cite{complextilingsjsl}),
it is natural to expect the quasi-periodicity function to inherit the
non-recursive properties of tilings. This is what had been proved
in~\cite{CervelleD04}.

In some particular cases it is easy to prove that this function is actually
computable. Consider a tileset such that any pattern that appears in a tiling
appears in every tiling;
%\footnote{In symbolic dynamics, this corresponds to minimal subshifts of finite type.}
in that case every tiling is quasi-periodic and the
quasi-periodicity function is the same for every tiling. Moreover there exists
an algorithm that decides if a pattern can appear in a tiling or not (this has
been proven by different ways, either by considering the fact that the first
order theory of the tileset is finitely axiomatizable and complete therefore
decidable~\cite{jacmodth} or by using a direct compactness
argument~\cite{hochmanrecsshift}).
Given this algorithm, it is easy to compute
the quasi-periodicity function (that does not depend on the tiling): for a given
$p$, compute all the $[-p;p]^2$ patterns that appear in a tiling and then
compute all the $[-n;n]^2$ patterns for $n\geq p$ until every $[-p;p]^2$ pattern
appears in every $[-n;n]^2$ pattern and output the smallest such $n$.

In the remainder of this paper, we improve this technique to obtain a less
restrictive condition on the tileset while proving that the quasi-periodicity
function is recursively bounded:

\begin{theorem}
\label{thm:main}
If a tileset (defined by $\ForbPat)$ admits only quasi-periodic tilings then
there exists a computable function $q: \N\to\N$ such that for any tiling
$\conf{c}$ of $\Tilings$, $\conf{c}$ has a quasi-periodicity
function bounded by $q$, \ie $\forall \conf{c}\in\Tilings, \forall n\in\N,
\qpfunc{c}(n)\leq q(n)$.
\end{theorem}

Note that this result is contrary to a result in~\cite{CervelleD04} stating that
there exists tilesets admitting only quasi-periodic tilings with
quasi-periodicity functions with no computable upper bound.
There is indeed a mistake in \cite{CervelleD04} that will be examined later.

\section{Computable bound on the quasi-periodicity function}
\label{sec:bound}

In this section we consider a tileset defined by a finite set of forbidden
patterns $\ForbPat$ such that every tiling by $\ForbPat$ is quasi-periodic.
The only hypothesis we have is the following: For any tiling
$\conf{c}\in\Tilings$ and for any pattern $\motif{P}$ that appears in
$\conf{c}$, there exists an integer $\qppat{P}{\conf{c}}$ such that any
$[-\qppat{P}{\conf{c}},\qppat{P}{\conf{c}}]^2$ pattern that appears in
$\conf{c}$ contains $\motif{P}$.
In order to prove Theorem~\ref{thm:main}, we first have to prove that there
exists a bound that does not depend on the tiling:

\begin{lemma}
\label{lemma:globalbound}
If a tileset $\ForbPat$ admits only quasi-periodic tilings then, for any pattern
$\motif{P}$ that appears in some tiling of $\Tilings$, there exists an integer
$n$ such that any tiling that contains $\motif{P}$ also contains $\motif{P}$ in
all its $[-n;n]^2$ patterns.

We define $\qppat{P}{\ForbPat}$ to be the smallest integer with this property.
\end{lemma}

Remark that the converse of this lemma is obviously true by definition: if for
any pattern there exists such an integer then all the tilings are
quasi-periodic.

\begin{proof}
Suppose this is not true: there exists a pattern $\motif{P}$ and a sequence
$(\conf{c}_n)_{n\in\N}$ of configurations that contain $\motif{P}$ and such that
$\conf{c}_n$ also contains a $[-n;n]^2$ pattern that does not contain
$\motif{P}$.

For a given $n$, consider $\motif{O}_n$, one of the largest square patterns of
$\conf{c}_n$ that does not contain $\motif{P}$. Since $\conf{c}_n$ is
quasi-periodic and contains $\motif{P}$ by hypothesis, there does not exist
arbitrary large square patterns that do not contain $\motif{P}$ and thus
$\motif{O}_n$ is well defined. Note that $\motif{O}_n$ is defined on at least
$[-n;n]^2$.
Since we supposed $\motif{O}_n$ of maximal size, there must be a pattern
$\motif{P}$ adjacent to it like depicted on Figure~\ref{fig:on}.

\begin{figure}[htb]
\begin{center}
\opt{dessinetikz}{\begin{tikzpicture}[scale=.4]

\filldraw[fill=\mblue] (0,3) rectangle +(5,5);
\filldraw[fill=\mgreen,draw=red, line width=3pt] (5,5) rectangle +(10,10);
\draw (0,0) grid +(15,15);
\node at (9.5,9.5) {$\motif{O}_n$};
\node at (2.5,5.5) {$\motif{P}$};
\end{tikzpicture}}\opt{pdftikz}{\includegraphics{dessins/on.pdf}}
\end{center}
\caption{$\motif{O}_n$ near $\motif{P}$.}
\label{fig:on}
\end{figure}

Now if we center our view on this $\motif{P}$ adjacent to $\motif{O}_n$, for
infinitely many $n$'s the largest part of $\motif{O}_n$ always
appears in the same quarter of plane (with origin $\motif{P}$).
Since $\motif{O}_n$ is defined on at least $[-n;n]^2$, by compactness we obtain
a tiling with $\motif{P}$ at its center and a quarter of plane without
$\motif{P}$. Such a tiling cannot be quasi-periodic.
\end{proof}

Lemma~\ref{lemma:globalbound} shows that if all the tilings that are valid for
$\ForbPat$ are quasi-periodic then there exists a global bound on the
quasi-periodicity function of any tiling: define $f(n)=\max\left\{\qppat{P}{\ForbPat},
\domain{P}=[-n;n]^2, \motif{P}\textrm{ appears in a tiling by }
\ForbPat\right\}$; for any tiling $\conf{c}\in\Tilings$ and any integer $n$, we
have $\qpfunc{c}(n)\leq f(n)$. The only part left in the proof of
Theorem~\ref{thm:main} is to prove that $f$ is computably bounded.

In a quasi-periodic tiling, if a pattern $\motif{P}$ defined on $[-n;n]^2$
appears in it then it must appear close to $\motif{P}$ (at distance less than
$f(n)+n$) in each of the four quarters of plane starting from the corners of
$\motif{P}$.
In general, we cannot compute whether a pattern will appear in some tiling or
not, however, we can compute whether a pattern is valid with respect to
$\ForbPat$.

\begin{lemma}
\label{lemma:localpattern}
If a tileset $\ForbPat$ admits only quasi-periodic tilings then, for any pattern
$\motif{P}$ defined on $[-n;n]^2$ that appears in some tiling of $\Tilings$,
there exists an integer $m$ such that any pattern $\motif{R}$ defined on
$[-n-m;n+m]^2$ that is valid with respect to $\ForbPat$ and contains
$\motif{P}$ at its center (\ie $\motif{R}_{|[-n;n]^2}=\motif{P}$) is such that
the four patterns \fourpat{R}{n}{m} all contain $\motif{P}$.

We define $\qppatloc{P}{\ForbPat}$ to be the smallest integer $m$ with this property.
\end{lemma}

Those four patterns may seem obscure at a first read, they are depicted on
Figure~\ref{fig:4pat}.

\begin{figure}[htb]
\begin{center}
\opt{dessinetikz}{\begin{tikzpicture}[scale=.3]

\filldraw[fill=\mblue] (8,8) rectangle +(4,4);
\filldraw[fill=\mgreen] (0,0) rectangle +(8,8);
\filldraw[fill=\mgreen] (0,12) rectangle +(8,8);
\filldraw[fill=\mgreen] (12,0) rectangle +(8,8);
\filldraw[fill=\mgreen] (12,12) rectangle +(8,8);

\draw (0,0) grid +(20,20);

\draw[<->] (21,12) -- node[right] {$m$} (21,20);
\draw[<->] (12,21) -- node[above] {$m$} (20,21);
\node at (10,10) {$\motif{P}$};

\node (text) at (-10,10) {$\motif{P}$ appears somewhere here};
\draw[->] (text) -- (16,16);
\draw[->] (text) -- (4,16);
\draw[->] (text) -- (4,4);
\draw[->] (text) -- (16,4);

\end{tikzpicture}}\opt{pdftikz}{\includegraphics{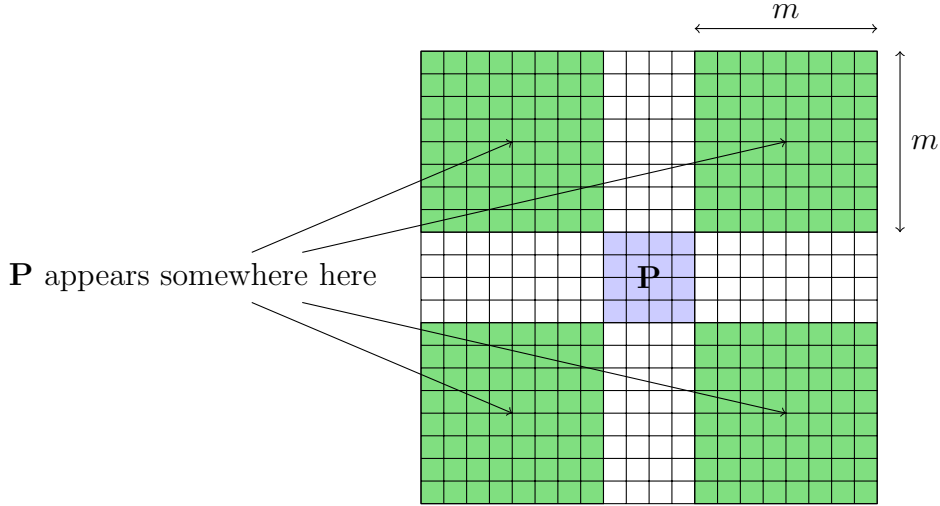}}
\end{center}
\caption{The four patterns in which we must find another occurrence of $\motif{P}$.}
\label{fig:4pat}
\end{figure}

\begin{proof}
For a given pattern $\motif{P}$, suppose that there exists no such
$m$. This means that there exist arbitrarily large $m$ and
valid patterns $\motif{R_m}$ (defined on $[-n-m;n+m]^2$) such that one of the
four patterns \fourpat{R_m}{n}{m} does not contain $\motif{P}$.

Without loss of generality, we can assume that this always happens in the same
quarter of plane.
By extracting a tiling centered on the pattern $\motif{P}$ at the center of
$\motif{R_m}$ (which we can do by compactness), there exists a tiling $\conf{c}$
of $\Tilings$ that contains $\motif{P}$ and a quarter of plane without
$\motif{P}$, contradicting the quasi-periodicity of $\conf{c}$.
\end{proof}

Note that the converse of Lemma~\ref{lemma:localpattern} is also true: if,
for any pattern $\motif{P}$, there exists such an $\qppatloc{P}{\ForbPat}$ then
all the tilings of $\Tilings$ are quasi-periodic.

\begin{lemma}
\label{lemma:ninfm}
If $\ForbPat$ is a tileset that allows only quasi-periodic tilings then, for any
pattern $\motif{P}$ defined on $[-p;p]^2$ that appears in some tiling of
$\Tilings$, we have:
\[
\qppat{P}{\ForbPat}\leq 2(\qppatloc{P}{\ForbPat} + p)
\]
\end{lemma}

\begin{proof}
Let $\conf{c}$ be a (quasi-periodic) tiling of $\Tilings$ that contains
$\motif{P}$ and a pattern $\motif{O}$ defined on $[-k;k]^2$ that does not
contain $\motif{P}$ with $k > 2(\qppatloc{P}{\ForbPat} + p)$.
Without loss of generality, we may assume that $\motif{O}$ is of maximal size.
That is, there is a pattern $\motif{P}$ adjacent to $\motif{O}$.
Let $\motif{R}$ be the pattern defined on
$[-p-\qppatloc{P}{\ForbPat};p+\qppatloc{P}{\ForbPat}]^2$ centered on the pattern
$\motif{P}$ adjacent to $\motif{O}$ in $\conf{c}$.
Since $k > 2(\qppatloc{P}{\ForbPat} + p)$ and $\motif{O}$ does not contain
$\motif{P}$, at least one of the four patterns
\fourpat{R}{p}{\qppatloc{P}{\ForbPat}} does not contain $\motif{P}$ as depicted on
Figure~\ref{fig:2mplusn}; since
$\motif{R}$ is a valid pattern with respect to $\ForbPat$, this contradicts the
definition of $\qppatloc{P}{\ForbPat}$.
\end{proof}

\begin{figure}[htb]
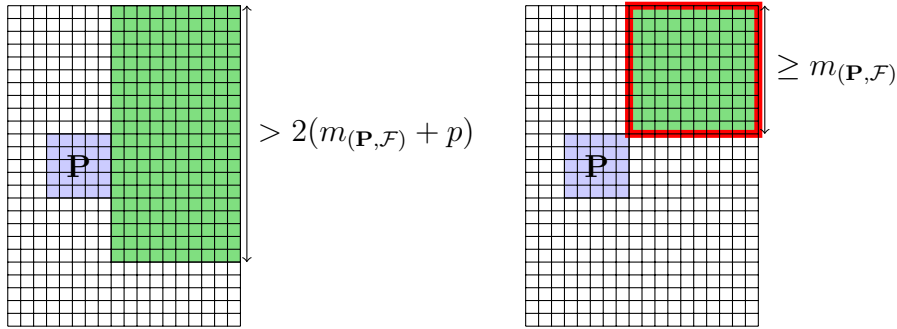

\begin{center}
\begin{tabular}{cc}
\opt{dessinetikz}{\begin{tikzpicture}[scale=.17]
\filldraw[fill=white,draw=white] (-4,-10) rectangle (20,16);
\filldraw[fill=\mblue] (0,0) rectangle +(5,5);
\filldraw[fill=\mgreen] (5,-5) rectangle +(10,20);
\draw (-3,-10) grid (15,15);
\node at (2.5,2.5) {$\motif{P}$};
\draw[<->] (15.5,-5) -- node[right] {$> 2(\qppatloc{P}{\ForbPat} + p)$} (15.5,15);
\end{tikzpicture}}\opt{pdftikz}{\includegraphics{dessins/2mplusn.pdf}}
&
\opt{dessinetikz}{\begin{tikzpicture}[scale=.17]

\filldraw[fill=white,draw=white] (-4,-10) rectangle (20,16);
\filldraw[fill=\mblue] (0,0) rectangle +(5,5);
\filldraw[fill=\mgreen,draw=red, line width=3pt] (5,5) rectangle +(10,10);
\draw (-3,-10) grid (15,15);
\node at (2.5,2.5) {$\motif{P}$};
\draw[<->] (15.5,5) -- node[right] {$\geq \qppatloc{P}{\ForbPat}$} (15.5,15);
\end{tikzpicture}}\opt{pdftikz}{\includegraphics{dessins/2mplusnimpos.pdf}}
\\
\end{tabular}
\end{center}
\caption{Bounding the size of the patterns not containing $\motif{P}$.}
\label{fig:2mplusn}
\end{figure}

Now that we have a bound that deals only about locally valid patterns instead of
patterns that appear in tilings (and therefore is computably checkable), we can
proceed to the proof of Theorem~\ref{thm:main}:

\begin{proof}[Proof of Theorem~\ref{thm:main}]
$\ForbPat$ is a tileset that admits only quasi-periodic tilings.
For an integer $n$, compute all the patterns $\motif{P_1},\ldots,\motif{P_k}$
defined on $[-n;n]^2$ that are valid for $\ForbPat$.

For each of these $\motif{P_j}$ use the following algorithm:
For each integer $i$, compute the set $\motif{R_1},\ldots,\motif{R_p}$ of
patterns defined on $[-i-n;i+n]^2$ that contain $\motif{P_j}$ at their center
and are valid with respect to $\ForbPat$.

\begin{enumerate}
\item If there is no such pattern $\motif{R}$, 
  claim that $\motif{P_j}$ cannot appear in any tiling by $\ForbPat$, and 
  define \eg $b_{\motif{P_j}}=0$. Then continue with $\motif{P_{j+1}}$
\label{enum:pasposs}
\item If all these patterns $\motif{R}$ restricted to either $[-n-i;-n]^2$,
$[-n-i;-n]\times[n;n+i]$, $[n;n+i]\times[-n-i;n]$ or $[n;n+i]^2$ all contain
$\motif{P}$ then define $b_{\motif{P_j}}=2(i+n)$ and continue with
$\motif{P_{j+1}}$\footnote{Remark that these patterns are exactly those depicted
in Figure~\ref{fig:4pat}.}.
\label{enum:upbound}
\end{enumerate}

For any pattern, one of these cases always happens: If $\motif{P_j}$ appears in
at least one tiling of $\Tilings$ then, by Lemma~\ref{lemma:localpattern}, for
$i=\qppatloc{P_j}{\ForbPat}$ we are in case~\ref{enum:upbound}. If $\motif{P_j}$
does not appear in any tiling of $\Tilings$ then case~\ref{enum:pasposs} must
happen, otherwise we would have arbitrary large extensions of $\motif{P_j}$ and
hence a tiling containing $\motif{P_j}$ by compactness. Note that we may halt in
case~\ref{enum:upbound} even if $\motif{P_j}$ does not appear in any tiling.

Now compute $q(n)=\max\left\{b_{\motif{P_j}}, \domain{P_j}=[-n;n]^2\right\}$.

For any tiling $\conf{c}\in\Tilings$ and any pattern $\motif{P}$ defined on
$[-n;n]^2$ that appears in $\conf{c}$ we have:
\[
\begin{array}{rcll}
\qppat{\motif{P}}{\conf{c}} & \leq & \qppat{\motif{P}}{\ForbPat} & \textrm{ by
definition of } \qppat{\motif{P}}{\ForbPat} \\
& \leq & 2(\qppatloc{\motif{P}}{\ForbPat} + n) & \textrm{ by Lemma~\ref{lemma:ninfm}}\\
& \leq & b_{\motif{P}}& \textrm{ by minimality of } \qppatloc{\motif{P}}{\ForbPat}\\
& \leq & q(n) & \textrm{ by definition of } q\\
\end{array}
\]

Therefore, for any configuration $\conf{c}$ and any integer $n$, we have 
$\qpfunc{c}(n)\leq q(n)$ and $q$ is the computable function that completes the
proof of Theorem~\ref{thm:main}.
\end{proof}

We remark that all the arguments used in the proofs of the lemmas 
involve only compactness and the fact that we can decide if a given pattern is
valid for $\ForbPat$. Hence, we may remove some restrictions on $\ForbPat$: $\Tilings$
is still compact if $\ForbPat$ is infinite and
we can still decide if a given pattern is valid for $\ForbPat$ when $\ForbPat$
is recursive.
Moreover, if $\ForbPat$ is recursively enumerable then there exists a recursive
set of patterns $\ForbPat'$ such that $\Tilings=\Tilings[\ForbPat']$: consider
the (computable) enumeration $f(0),f(1),\ldots$ of $\ForbPat$; when enumerating
$f(i)$, we can compute an integer $n$ such that all the previously enumerated
patterns are defined on a domain included in $[-n;n]^2$; then we enumerate all
the extensions of $f(i)$ defined on $[-n-1;n+1]^2 \cup \domain{f(i)}$.
This enumeration enumerates a new set of patterns $\ForbPat'$ that is now
recursive since they are enumerated by increasing sizes. It is straightforward
that $\Tilings=\Tilings[\ForbPat']$.
We conclude that $\ForbPat$ needs not to be finite in order for
Theorem~\ref{thm:main} to be valid but we may assume that it is only recursively
enumerable.
Sets of tilings with a recursively enumerable set of forbidden patterns are
usually called \emph{effective subshifts} in the
literature~\cite{compuniveffsymb,hochmanrecsshift,hochmanuniv,mathieunathaliesoficeff,drseffsofic}
or also $\Pi_1^0$ subshifts~\cite{simpsonmedv,millereffsub} and are a
special case of effectively closed sets as studied in computable analysis (see
\eg \cite{companalbook})\footnote{The definitions are usually given in dimension
one, \ie for (bi-)infinite, words even though they are the same for
multi-dimensional configurations.}.

\section{Large quasi-periodicity functions}
\label{sec:lowbound}

In this section we prove that we can construct tilesets whose every
quasi-periodic tiling has a large quasi-periodicity function.
%The idea is as follows.
We start from a 1-dimensional effective subshift
$\sshift{X}$ over an alphabet $\Sigma$ and then build an effective subshift
over the alphabet $\Sigma \times \{0,1\}$, and the complexity of the
quasi-periodicity function will come from the top layer. For this,
consider all occurrences of a word $u$ in the subshift $\sshift{X}$. There
are infinitely many of them, so the top layer restricted to occurrences of $u$ 
will contain a bi-infinite word over $\{0,1\}$.
If we can find infinitely many words in the subshift $\sshift{X}$ so that
occurrences of different words do not somehow overlap in a configuration $c$,
then this would give us an infinite number of bi-infinite words within
a single configuration $c$, in which we could code something.

The following lemma tells us how to find such words in the general case of minimal
effective subshifts; a \emph{minimal subshift} is a subshift in which every
pattern that appears in a configuration appears in every configuration, or
equivalently, a subshift that does not admit a proper non-empty subshift.
In this case, all configurations are of course quasi-periodic.

\begin{lemma}
\label{lemma:sequnweak}
For any (non-empty) $1-$dimensional minimal effective subshift
$\sshift{X}\subseteq\Sigma^{\Z}$ that has no periodic
configuration there exists a computable sequence $(u_n)_{n\in\N}$ of words in
the language of $\sshift{X}$ such that no $u_n$ is prefix of another one.
\end{lemma}
%So the occurrences of the first characters of each $u_n$ will not overlap.
\begin{proof}
We build recursively a sequence $(u_0,\ldots, u_n)$ and a word $v_n$ 
such that the set $\{u_k, k \leq n\} \cup \{v_n\}$ is prefix-free.
For $n=0$, take two different letters in $\Sigma$ ($|\Sigma|>1$ comes from the
hypothesis as $\sshift{X}$ is non-empty and does not contain any periodic
configuration).

Now suppose we obtain $(u_0, \ldots u_n)$ and $v_n$.
Since $\sshift{X}$ is supposed to be minimal, $v$ appears in an uniformly
recurrent way in a configuration of $\sshift{X}$ and since
$\sshift{X}$ contains no periodic configuration,
there exists two different right-extensions of $v$: $w$ and $w'$ of the same
length. Taking $u_{n+1} = w$ and $v_{n+1}=w'$ ends the recurrence.
\end{proof}

To obtain our theorem, we will need a subshift $\sshift{X}$ for which we
control precisely the sequence $u_n$.

\begin{lemma}
\label{lemma:sequn}
There exists a (non-empty) 1-dimensional minimal effective subshift $\sshift{X}$
and a computable sequence $(u_n)_{n \in \N}$ of words in the
language of $\sshift{X}$ so that $|u_n|\leq n$ and no $u_n$ is prefix of
another one.
\end{lemma}
\begin{proof}
We will use a construction based on Toeplitz words.
Let $p$ be an integer.
For an integer $n$, denote by $\phi_p(n)$ the first non-zero digit in the
writing of $n$ in base $p$, \eg $\phi_3(15) = 2$.

Let $w_p = \phi_p(1) \phi_p(2) \dots $.
For example $w_4 = 12311232123312311231123\ldots$.%21233\dots$.

Now let $\sshift{X}_p$ be the shift of all configurations $c$ so that all words of
$c$ are words of $w_p$. Note that any word of size $n$ appearing in
$w_p$ appears at a position less than $p^n$ so that $\sshift{X}_p$ is an effective
subshift.

Now the following statements are clear:
\begin{itemize}
	\item For every word $w$ in $w_p$, there exists $k$ so that 
	  for every configuration $c \in \sshift{X}_p$,  $w$ appears periodically in $c$ of period
	  $p^k$  ($w$ might appear in some other places)
	\item $\sshift{X}_p$ is minimal (a consequence of the previous statement)
\end{itemize}	

If $u_1$ and $u_2$ are two words over $\Sigma_1$ and $\Sigma_2$ of the same size,
we write $u_1 \otimes u_2$
for the word over $\Sigma_1 \times \Sigma_2$ whose $i$th projection is $u_i$
($i\in\left\{1,2\right\}$).

Now let $\sshift{X} = \sshift{X}_7 \otimes \sshift{X}_8$. $\sshift{X}$ is a shift, 
and  $\sshift{X}$ is minimal\footnote{Note that the Cartesian product of two
minimal shifts is not always minimal \cite{Salimov}.}:
If $c_1 \otimes c_2 \in \sshift{X}_7 \otimes \sshift{X}_8$ and 
$u_1$ and $u_2$ are two patterns resp. of $w_7$ and $w_8$ of the same size,
then $u_1$ appears periodically in $c_1$ of period $7^{k_1}$ and $u_2$
appears periodically in $c_2$ of period $8^{k_2}$. As these two numbers are
relatively prime, there exists a common position $i$ so that $u_1$
(resp. $u_2$) appears in position $i$ in $c_1$ (resp $c_2$), so that
$u_1 \otimes u_2$ appears in $c_1 \otimes c_2$.

Now we can find the sequence $u_n$.

Let $u$ be a word in $\{5,6,7\}^\star\{1,2,3,4\}$.
We define $f_8(u)$ inductively as follows:
\begin{itemize}
	\item If $|u| = 1$, then $f_8(u) = u$.
	\item If $u = xu_1$ then let $v = f_8(u_1)$ and $n$ be the length of $v$.
	  \begin{itemize}
		  \item If $x = 5$ then $f_8(u) = 567v_11234567v_21234567v_31\dots v_n1234$
		  \item If $x = 6$ then $f_8(u) = 67v_11234567v_21234567v_31\dots v_n12345$
		  \item If $x = 7$ then $f_8(u) = 7v_11234567v_21234567v_31\dots v_n123456$
	  \end{itemize}	  
\end{itemize}	
Now it is clear that each $f_8(u)$ is in $w_8$ and by a straightforward
induction, no $f_8(u)$ is prefix of another.
Let $S_8 = \{f_8(u) | u \in \{5,6,7\}^\star\{1,2,3,4\} \}$
Note that $f_8(u)$ is of length $8^{|u|-1}$.
In particular we have $4\times 3^{n-1}$ words of length $8^{n-1}$ in $S_8$.

We do the same with $w_7$, with words $u \in \{4,5,6\}^\star\{1,2,3\}$, to
obtain a set $S_7$ containing $3\times 3^{n-1}$ words of length $7^{n-1}$.
We can always enlarge all words in $S_7$ to obtain a set $S'_7$ containing
$3\times 3^{n-1}$ words of length $8^{n-1}$.

Now take $S = S'_7 \otimes S_8$. This set contains $12\times 9^{n-1} > 8^n$ words
of size $8^{n-1}$ for each $n$ and no word of $S$ is prefix of one another.
Now an enumeration in increasing order of $S$ gives the sequence
$(u_n)_{n\in\N}$.

The whole construction is clearly effective.
\end{proof}

\begin{theorem}
\label{thm:1deffnonbound}
Given a partial computable function $\varphi$, there exists a $1-$dimensional
effective subshift $\sshift{X}_{\varphi}$ such that any quasi-periodic
configuration $\conf{c}$ in $\sshift{X}_{\varphi}$ has a quasi-periodicity
function $\qpfunc{c}$ such that $\qpfunc{c}(n)\geq\varphi(n)$ when $\varphi(n)$ is
defined.
\end{theorem}

\begin{proof}
Consider the subshift $\sshift{X}$ and the computable sequence
$(u_n)_{n\in\N}$ that are given by Lemma~\ref{lemma:sequn}.
Since Lemma~\ref{lemma:sequn} ensures that $|u_n|\leq n$, a sequence
$(u_n)_{n\in\N}$ with the additional property that $|u_n|=n$ is also computable
since we can compute an extension of the words $u_n$ in $\sshift{X}$ since it is
minimal and effective and the prefix-free property is retained while taking
extensions. We assume this additional property in this proof.

Let $\Sigma'=\Sigma\times\left\{0,1\right\}$. We define
$\sshift{X}_{\varphi}$ as a subshift of $\sshift{X}\times\left\{0,1\right\}^{\Z}$.

Compute in parallel all the $\varphi(n)$.
When $\varphi(n)$ is computed we add the following additional constraints:
On the $\left\{0,1\right\}$ layer of $\Sigma'$ we force a $1$ to appear on the
first letter of $u_n$ once every $\varphi(n)+1$ occurrences of $u_n$, the
first letter of all other occurrences of $u_n$ being $0$.
There is no ambiguity since no $u_n$ is prefix of another one.
This defines $\sshift{X}_{\varphi}$ as an effective subshift since $\sshift{X}$
is effective and $(u_n)_{n\in\N}$ is computable.

Every $u_n$ appears in every configuration of $\sshift{X}$ since it is minimal.
If $\varphi(n)$ is defined, then every $u_n$ with a $1$ on the
$\left\{0,1\right\}$ layer appears exactly every $\varphi(n)$ occurrences of $u_n$'s with a $0$ on its $\left\{0,1\right\}$ layer
in every configuration of $\sshift{X}_{\varphi}$.
Therefore, for any quasi-periodic
configuration $\conf{c}$ of $\sshift{X}_{\varphi}$ we have that
$\qpfunc{c}(n)\geq\varphi(n)$ where $\varphi(n)$ is defined which completes the proof.
\end{proof}

\iffalse
% On a pas ça, mais peut-on l'avoir ?
Remark that if $\varphi$ is computable (\ie it is total) then
$\sshift{X}_{\varphi}$ contains only quasi-periodic configurations.
\fi

\begin{cor}
	There exists a 1-dimensional effective subshift $\sshift{X}$ such that 
every quasi-periodic configuration $c$ in $\sshift{X}$ has a
quasi-periodicity function which is not bounded by any computable function.
\end{cor}	
\begin{proof}
Let $(\varphi_n)_{n\in\N}$ be an effective enumeration of partial computable
functions.

Let $\varphi(n)=\varphi_n(n)+1$; $\varphi$ is also a partial computable function;
we can therefore find an effective one dimensional subshift
$\sshift{X}_{\varphi}\subseteq\Sigma^{\Z}$ via Theorem~\ref{thm:1deffnonbound}
such that any quasi-periodic configuration $\conf{c}$ of $\sshift{X}_{\varphi}$
is such that $\qpfunc{c}\geq\varphi$ where $\varphi$ is defined, hence
$\qpfunc{c}$ is not recursively bounded.
\end{proof}

\iffalse
\begin{cor}
	There exists a 1-dimensional effective subshift $\sshift{X}$ such that 
every quasi-periodic configuration $c$ in $\sshift{X}$ has a
quasi-periodicity function that grows faster than any computable function.
\end{cor}	
\begin{proof}

comment l'avoir ?

\end{proof}
\fi
\begin{theorem}
\label{thm:qpfuncunbound}
There exists a tileset such that every quasi-periodic tiling has a
quasi-periodicity function that is not recursively bounded.
\end{theorem}

\begin{proof}
Take the effective $1-$dimensional subshift of the previous corollary (as a
subshift of $\Sigma^{\Z}$): $\sshift{X}_{\varphi}$.
There exists a set of tilings (or $2-$dimensional SFT)
$\sshift{X}_{\varphi}^2\subseteq (Q\times\Sigma)^{\Z^2}$ encoding
it~\cite{mathieunathaliesoficeff,drseffsofic} in the following way:

In any configuration of $\sshift{X}_{\varphi}^2$, the rows of the $\Sigma-$layer
are identical, that is, if we write this configuration as $c_{Q}\times
c_{\Sigma}\in Q^{\Z^2}\times\Sigma^{\Z^2}$, for any $i,j$ in $\Z$,
$c_{\Sigma}(i,j)=c_{\Sigma}(i,j+1)$. Moreover, the projection:
\[
\begin{array}{rrcl}
p:&(Q\times\Sigma)^{\Z^2}&\to&\Sigma^{\Z}\\
& c_{Q}\times c_{\Sigma} & \to &
\begin{array}{rcl}
\Z&\to&\Sigma\\
n&\to&c_{\Sigma}(n,0)\\
\end{array}
\\
\end{array}
\]
of $\sshift{X}_{\varphi}^2$ is exactly
$\sshift{X}_{\varphi}$ (\ie $p(\sshift{X}_{\varphi}^2)=\sshift{X}_{\varphi}$).
Since the configurations of $\sshift{X}_{\varphi}^2$ are the Cartesian product
of a construction layer (the $Q^{\Z^2}$ part) and the effective $1-$dimensional
subshift $\sshift{X}_{\varphi}$ repeated on the rows, the quasi-periodicity
function of any quasi-periodic configuration of $\sshift{X}_{\varphi}^2$ is
greater or equal to the quasi-periodicity function of the quasi-periodic 1-dimensional configuration it represents.
\end{proof}

Note that quasi-periodicity configurations obtained in
the constructions in \cite{mathieunathaliesoficeff,drseffsofic} are rather
benign.
If we start from a $1-$dimensional quasi-periodic configuration $\conf{c}$, then
the quasi-periodic tilings $\conf{x}$ that are projected onto $\conf{c}$ have a
quasi-periodicity function that is computable knowing the quasi-periodicity
function of $\conf{c}$.

\section{Note}

Theorem~$9$ in \cite{CervelleD04} stated the contrary of Theorem~\ref{thm:main}:
``there exists a tileset such that all its tilings are quasi-periodic and none
of its quasi-periodicity function is computably bounded''.
Besides some errors that can be easily corrected, there is a big
problem in the construction they claim to give.
They encode $K$, a recursively enumerable but not recursive set, in every tiling
in a way such that if $i\in K$ then it must appear in every tiling in a pattern
of size $g(i)$ where $g$ is a computable function.
This property allows by itself to decide $K$:
For an integer $i$, compute $g(i)$ and all the possible encodings of $i$ if it
were to appear in a tiling; patterns that do not appear in a tiling of the plane
are recursively enumerable\footnote{Simply try to tile arbitrary big patterns
around it and if it is not possible claim that the pattern does not appear in a
tiling.} and thus, when we have enumerated all the patterns coding $i$ we know
that $i\not\in K$. Since $K$ is supposed recursively enumerable, this allows to
decide $K$.

\bibliographystyle{plain}
\bibliography{biblio/article,biblio/ca-faq,biblio/jac08,biblio/stacs08}
\end{document}